\documentclass[sn-mathphys-num]{sn-jnl}

\usepackage{graphicx}
\usepackage{multirow}
\usepackage{amsmath,amssymb,amsfonts}
\usepackage{amsthm}
\usepackage{mathrsfs}
\usepackage[title]{appendix}
\usepackage{xcolor}
\usepackage{textcomp}
\usepackage{manyfoot}
\usepackage{booktabs}
\usepackage{algorithm}
\usepackage{algorithmicx}
\usepackage{algpseudocode}
\usepackage{listings}
\usepackage{tikz}
\usepackage{istgame}
\usepackage{subfigure}
\usepackage{physics} 
\usepackage{float}
\usepackage{etoolbox}
\usepackage{natbib}

\theoremstyle{thmstyleone}%
\newtheorem{theorem}{Theorem}
\theoremstyle{thmstyletwo}%

\theoremstyle{thmstylethree}%

\begin{document}

\title[Grand Angular Momentum]{Classical Grand Angular Momentum in N-Body Problems}

\author*[1,2]{\fnm{Zhongqi} \sur{Liang}}\email{zhongqi.liang@stonybrook.edu}

\author[1,3]{\fnm{Jes\'us} \sur{P\'erez-R\'ios}}\email{jesus.perezrios@stonybrook.edu}

\affil[1]{\orgdiv{Department of Physics and Astronomy}, \orgname{Stony Brook University}, \orgaddress{\city{Stony Brook}, \postcode{11790}, \state{NY}, \country{US}}}


\abstract{The concept of grand angular momentum is widely used in the study of N-body problems quantum mechanically. Here, we applied it to a classical analysis of N-body problems. Utilizing the tree representation for Jacobi and hyperspherical coordinates, we found a decomposition of its magnitude into magnitudes of one-body angular momenta in three dimensions. We generalized some results from the two-body case and derived a general expression for the scattering angle in N-body problems.}

\keywords{N-Body Problem, Classical Mechanics, Angular Momentum, Scattering}



\maketitle
\section{Introduction}\label{sec1}

    The N-body problem concerns the time evolution of an isolated system consisting of $N$ objects in a three-dimensional space with only internal interactions that could be modeled by a potential $V$ dependent solely on the distances between them. It is ubiquitous in physics -- first studied in large-scale astronomical contexts as many problems in physics by illustrious names such as Euler and Lagrange \cite{Bann2022, Arnold2006} and later, especially after the advances of quantum mechanics, in relevant scenarios from atomic and nuclear physics ~\cite{Thomas,Faddeev1965,Greene2017,PR2014,PR2020,Xiao2013, Xiao2015,Lomb2007,Aqui1986,Aqui1987,EE2018}. At the same time, it is also well known that an exact solution does not exist for $N > 2$ except for a few special cases. Nevertheless, one could still greatly simplify the general problem by making appropriate coordinate transformations and applying conservation laws for any $N$. 

    One set of such coordinates is the Jacobi coordinates, which, after enforcing momentum conservation and neglecting the trivial center of mass motion, can reduce it to an effective $(N-1)$-body problem. The various ways of forming these virtual bodies could be graphically represented by binary Jacobi trees, \cite{Lim1989, Lim1991, Lim2020} which will be introduced in Part \ref{Jacobi}. By a mass-weighted rescaling of the coordinates associated with each of the $(N-1)$ bodies, it can be further reduced to a problem of one virtual body moving in $3N - 3$ dimensions~\cite{Smith1960,PR2014,Greene2017,PR2020}. This higher-dimensional space is usually described with hyperspherical coordinates, first introduced in nuclear physics, ~\cite{Hyperspherical_first} and showing interesting results in the quantum realm, such as the existence of Efimov states~\cite{Efimov1970}. There are also different definitions of hyperspherical coordinates, which can again be represented by binary trees. The details will be laid out in Part \ref{hpsph}. Moreover, the grand angular momentum, as an extension of the concept of angular momentum to higher dimensions, is a conserved magnitude in many scattering scenarios and, as a result, has well-defined eigenvalues and associated eigenfunctions at the quantum level~\cite{Avery2018, Aqui1986, Smith1960}. In fact, the concept of grand angular momentum has been more thoroughly studied in a quantum mechanical context, while its applications in the framework of classical mechanics have received less attention.

    Here, in Part \ref{grandL}, we bridge this gap with the help of graphical representations of the coordinates. We prove that one can read directly from a tree an expression for the magnitude of grand angular momentum. Moreover, equipped with this knowledge, a decomposition of the grand angular momentum into regular angular momentum in their magnitudes can be achieved, and a second-order differential equation that constrains the evolution of the $3N - 4$ hyperangles is found.

\section{Jacobi Trees} \label{Jacobi}

    For the two-body problem in a three-dimensional space, the only setup that emits a general analytical solution of particle trajectories within the N-body context, the Hamiltonian is given by
    \begin{equation}
    \label{eq1}
        H = \frac{\boldsymbol{p}_1^2}{2m_1} + \frac{\boldsymbol{p}_2^2}{2m_2} + V(\boldsymbol{r}_2 - \boldsymbol{r}_1),
    \end{equation}
    which describes a system made up of masses $m_1$ and $m_2$ at positions $\boldsymbol{r}_1$ and $\boldsymbol{r}_2$ with conjugate momenta $\boldsymbol{p}_1$ and $\boldsymbol{p}_2$, respectively, interacting via the interaction potential $V(\boldsymbol{r}_2 - \boldsymbol{r}_1)$. Using the following coordinate transformation
    \begin{eqnarray}
        \boldsymbol{R}_{\text{CM}} &=& \boldsymbol{R}_{12} = \frac{m_1 \boldsymbol{r}_1 + m_2 \boldsymbol{r}_2}{m_1 + m_2} \\
        \boldsymbol{\rho} &=& \boldsymbol{\rho}_1 = \boldsymbol{r}_2 - \boldsymbol{r}_1,
    \end{eqnarray}
    the Hamiltonian~(\ref{eq1}) reads
    \begin{equation}
        H = \frac{\boldsymbol{P}^2}{2\mu} + \frac{\boldsymbol{P}_{\text{CM}}^2}{2M} + V(\boldsymbol{\rho}),
    \end{equation}
    where $\mu = \mu_{1,2} = m_1 m_2 / (m_1 + m_2)$ is the two-body reduced mass and  $M = M_{12} = m_1 + m_2$ is the total mass of the system. Due to the nature of the interaction potential, $\boldsymbol{R}_{\text{CM}} = \boldsymbol{R}_{12}$, does not appear in the interaction potential. Hence, the center of mass momentum, $\boldsymbol{P}_{\text{CM}}$, is a conserved quantity. Therefore, the two-body problem reduces to a single particle with mass $\mu = \mu_{1,2}$, placed at $\boldsymbol{\rho} = \boldsymbol{\rho}_1$ with momentum $\boldsymbol{P} = \mu \dot{\boldsymbol{\rho}}$. 

    For $N=3$, after neglecting the center of mass motion, two virtual bodies are required to specify the dynamics. The first one describes the relative motion between $m_1$ and $m_2$, the same as in the case of the two-body problem. The second is associated with $\boldsymbol{\rho}_2 = \boldsymbol{r}_3 - \boldsymbol{R}_{12} $, joining the center of mass of the first two particles with the position of the third. In these new coordinates, the Hamiltonian reads \begin{align}
        H &= \frac{\boldsymbol{p}_1^2}{2m_1} + \frac{\boldsymbol{p}_2^2}{2m_2} + \frac{\boldsymbol{p}_3^2}{2m_1} + V(\boldsymbol{r}_2- \boldsymbol{r}_1, \boldsymbol{r}_3- \boldsymbol{r}_2,\boldsymbol{r}_1- \boldsymbol{r}_3) \\
        &= \frac{\boldsymbol{P}_1^2}{2\mu_{1,2}} + \frac{\boldsymbol{P}_2^2}{2\mu_{12,3}} + \frac{\boldsymbol{P}_{CM}^2}{2M} + V(\boldsymbol{\rho}_1, \boldsymbol{\rho}_2),
    \end{align}
    with
    \begin{eqnarray}
     \mu_{12,3} = \frac{(m_1 + m_2) m_3}{(m_1 + m_2) + m_3}, 
    \end{eqnarray}
    and $M = M_{123} = m_1 + m_2 + m_3$. 
   
    It is possible to generalize this approach to any $N$, building Jacobi coordinates recursively as
    \begin{eqnarray} 
    \label{2}
        M_{12...j} = \sum_{i=1}^j m_i, \\
        \mu_{12...j-1,j} = \frac{M_{12...j-1} m_j}{M_{12...j}}, \\
        \boldsymbol{R}_{12...j} = \frac{\sum_{i=1}^j m_i \boldsymbol{r}_i}{M_{12...j}}, \\ \boldsymbol{\rho}_j = \boldsymbol{r}_{j+1} - \boldsymbol{R}_{12...j},
    \end{eqnarray}
    for $j = 1, 2, ..., N - 1$. The Hamiltonian is given by
    \begin{equation}
        H = \sum_{i=1}^{N-1} \frac{\boldsymbol{P}_i^2}{2\mu_{12...i,i+1}} + V(\boldsymbol{\rho}_1,...,\boldsymbol{\rho}_{N-1}),
    \end{equation}
    after removing the center-of-mass term. If the system contains an infinitely massive object $m_1$, then the Jacobi coordinates constructed as such simply remove the degrees of freedom associated with $m_1$ and the $N-1$ bodies concerned are the remaining $N-1$ masses with $\boldsymbol{r}_j - \boldsymbol{r}_1, j=2,...,N$ as their coordinates. 
    
    Further, by introducing the N-body reduced mass~\cite{Smith1960} 
    \begin{equation} \label{1}
        \mu \equiv \left(\frac{\prod_{i=1}^N m_i}{\sum_{i=1}^N m_i}\right)^{1/(N-1)},
    \end{equation}    
    it is possible to define mass-weighted Jacobi coordinates as 
    \begin{equation} \label{3}
        \boldsymbol{\rho}_{\text{MW},i} = \sqrt{\frac{\mu_{12...i,i+1}}{\mu}} \boldsymbol{\rho}_i. 
    \end{equation}
    Next, stacking them on top of each other, namely
    \begin{equation}
        \boldsymbol{\rho}_{\text{MW}} \equiv \begin{pmatrix}
            \boldsymbol{\rho}_{\text{MW},1} \\ \boldsymbol{\rho}_{\text{MW},2} \\ . \\ . \\ . \\\boldsymbol{\rho}_{\text{MW},N-1}
        \end{pmatrix},
    \end{equation}
    the Hamiltonian is reduced to
    \begin{equation} \label{4}
        H = \frac{\boldsymbol{P}_{\text{MW}}^2}{2 \mu} + V(\boldsymbol{\rho}_{\text{MW}}), 
    \end{equation}
    with $\boldsymbol{\rho}_{\text{MW}}$ and its conjugate momentum $\boldsymbol{P}_{\text{MW}} = \mu \boldsymbol{\dot{\rho}}_{\text{MW}}$ vectors in $3N - 3$ dimensions. The degrees of freedom are thus repackaged and the system can be interpreted as one body $\mu$ at $\boldsymbol{\rho}_{\text{MW}}$ moving with momentum $\boldsymbol{P}_{\text{MW}}$ in a $(3N-3)$-dimensional space. \cite{PR2020, Greene2017}

    However, as $N$ increases, there also appear more ways of forming such $N - 1$ virtual bodies. For $N = 2$, there is no ambiguity. For $N = 3$, it is still unique up to labeling. When $N = 4$, in addition to the recursive way presented above, it is also permitted to have three virtual bodies as $\mu_{1,2}$, $\mu_{3,4}$, and $\mu_{12,34} = (m_1 + m_2)(m_3 + m_4)/M$, as shown in Fig.~\ref{fig:4-body}. Each particular choice of Jacobi coordinates has attached to it a binary Jacobi tree, as presented in panels (b) and (d) of Fig.~\ref{fig:4-body}. 

    \begin{figure}[b]
        \centering
        \subfigure[]{
            \scalebox{0.8}{
                \begin{tikzpicture}
                    [real/.style={circle,draw=black,fill=black,thick,inner sep=2pt},
                    virtual/.style={circle,draw=black,thick,inner sep=2pt}]
                    \node (m1) at (-2.5, 0.5) [real] [label=above:$m_1$]{};
                    \node (m2) at (0.5, -2.5) [real] [label=below:$m_2$] {};
                    \node (m3) at (3, -0.5) [real] [label=below:$m_3$] {};
                    \node (m4) at (3, 2) [real] [label=above:$m_4$] {};
                    \node (mu12) at (-0.5, -1.5) [virtual] [label=210:$\mu_{1,2}$] {};
                    \node (mu123) at (9/10, -11/10) [virtual] [label=300:$\mu_{12,3}$] {};
                    \node (mu1234) at (27/16, 1/16) [virtual] [label=135:$\mu_{123,4}$] {};
                    \draw (m1) -- (mu12) [thick];
                    \draw [->] (mu12) -- (m2) [thick];
                    \draw (mu12) -- (mu123) [thick];
                    \draw [->] (mu123) -- (m3) [thick];
                    \draw (mu123) -- (mu1234) [thick];
                    \draw [->] (mu1234) -- (m4) [thick];
                \end{tikzpicture}
            }
        }
        \hspace{1cm}
        \subfigure[]{
            \scalebox{0.8}{
                \begin{istgame}
                    \setistgrowdirection'{north}
                    \setistNullNodeStyle{0pt}
                    \setistHollowNodeStyle{5pt}
                    \setistSolidNodeStyle{5pt}
                    \xtShowTerminalNodes
                    \xtdistance{10mm}{20mm}
                    \istroot(0)[hollow node]<[yshift=-20pt]>{$\mu_{123,4}$}
                        \istb \istb
                    \endist
                    \istroot(1)(0-1)[hollow node]<[yshift=-15pt, xshift=-15pt]>{$\mu_{12,3}$}
                        \istb \istb
                    \endist
                    \istroot(2)(0-2)[null node]{}
                        \istbm \istb
                    \endist
                    \istroot(11)(1-1)[hollow node]<[yshift=-15pt, xshift=-15pt]>{$\mu_{1,2}$}
                        \istbt{}[al]{m_1} \istbt{}[al]{m_2}
                    \endist
                    \istroot(12)(1-2)[null node]{}
                        \istbm \istbt{}[al]{m_3}
                    \endist
                    \istroot(22)(2-2)[null node]{}
                        \istbm \istbt{}[al]{m_4}
                    \endist
                \end{istgame}
            }
        }
        \vfill
        \subfigure[]{
            \scalebox{0.8}{
                \begin{tikzpicture}
                    [real/.style={circle,draw=black,fill=black,thick,inner sep=2pt},
                    virtual/.style={circle,draw=black,thick,inner sep=2pt}]
                    \node (m1) at (-2.5, 0.5) [real] [label=above:$m_1$]{};
                    \node (m2) at (0.5, -2.5) [real] [label=below:$m_2$] {};
                    \node (m3) at (3, -0.5) [real] [label=below:$m_3$] {};
                    \node (m4) at (3, 2) [real] [label=above:$m_4$] {};
                    \node (mu12) at (-0.5, -1.5) [virtual] [label=210:$\mu_{1,2}$] {};
                    \node (mu34) at (3, 1) [virtual] [label=0:$\mu_{3,4}$] {};
                    \node (mu1234) at (27/16, 1/16) [virtual] [label=120:$\mu_{12,34}$] {};
                    \draw (m1) -- (mu12) [thick];
                    \draw [->] (mu12) -- (m2) [thick];
                    \draw (mu12) -- (mu1234) [thick];
                    \draw [->] (mu1234) -- (mu34) [thick];
                    \draw (m3) -- (mu34) [thick];
                    \draw [->] (mu34) -- (m4) [thick];
                \end{tikzpicture}
            }
        }
        \hspace{1cm}
        \subfigure[]{
            \scalebox{0.8}{
                \begin{istgame}
                    \setistgrowdirection'{north}
                    \setistNullNodeStyle{0pt}
                    \setistHollowNodeStyle{5pt}
                    \setistSolidNodeStyle{5pt}
                    \xtShowTerminalNodes
                    \xtdistance{10mm}{20mm}
                    \istroot(0)[hollow node]<[yshift=-20pt]>{$\mu_{12,34}$}
                        \istb \istb
                    \endist
                    \istroot(1)(0-1)[null node]{}
                        \istb \istbm
                    \endist
                    \istroot(2)(0-2)[null node]{}
                        \istbm \istb
                    \endist
                    \istroot(11)(1-1)[hollow node]<[yshift=-15pt, xshift=-15pt]>{$\mu_{1,2}$}
                        \istbt{}[al]{m_1} \istbt{}[al]{m_2}
                    \endist
                    \istroot(22)(2-2)[hollow node]<[yshift=-15pt, xshift=15pt]>{$\mu_{3,4}$}
                        \istbt{}[al]{m_3} \istbt{}[al]{m_4}
                    \endist
                \end{istgame}
            }
        }
        \caption{Visualization of Jacobi coordinates (a) and (c), together with their respective tree representations (b) and (d) for the 4-body problem. Solid circles represent physical bodies. Hollow circles are for virtual bodies. Unweighted distance vectors are the arrowed lines.}
        \label{fig:4-body}
    \end{figure}
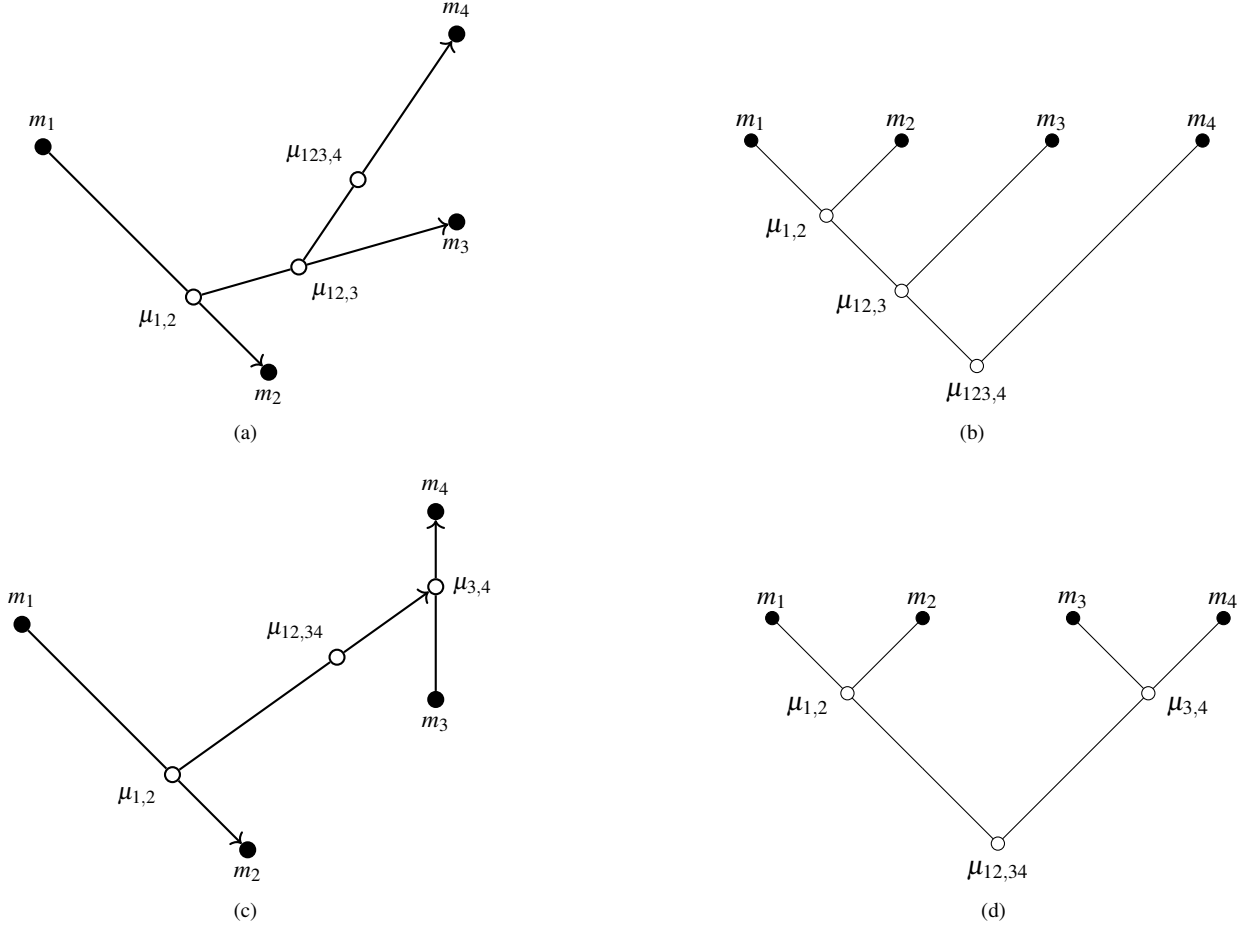

    The correspondence between the coordinates and their tree representation is straightforward. Each leaf represents a physical body and each node a virtual body of mass $\mu_{L,R}$, where $L = \{l_1, l_2, ..., l_m\}$ and $R = \{r_1, r_2, ..., r_n\}$ are collective indices of $m$ physical masses joining the node from the left and $n$ from the right, respectively, and the $\boldsymbol{\rho}_{L,R}$ associated is given by $\boldsymbol{R}_{R} - \boldsymbol{R}_{L}$, the vector from the center of mass of all masses on the left to that of all masses on the right. More explicitly, for each node, 
    \begin{eqnarray} \label{5}
        M_L = M_{l_1 l_2...l_m} = \sum_{i=1}^m m_{l_i}, \\
        \mu_{L,R} = \frac{M_L M_R}{M_L + M_R}, \\
        R_L = R_{l_1 l_2...l_m} = \frac{\sum_{i=1}^m m_{l_i} \boldsymbol{r}_{l_i}}{M_{l_1 l_2...l_m}}, \\
        \boldsymbol{\rho}_{L,R} = \boldsymbol{R}_R - \boldsymbol{R}_L,
    \end{eqnarray}
    could be defined, and there exist $N - 1$ such nodes \cite{Lim1989, Lim1991}. By the same procedures described by Eq.~\ref{3}, again a one-body problem in $3N - 3$ dimensions is at hand.

    Beyond this point in all later sections, the ``MW" subscript will be dropped but implicitly assumed so that every vector, unless otherwise noted, will be mass-weighted.

    \section{Hyperspherical Trees} \label{hpsph}

    After converting an N-body problem in 3 dimensions to an effective one-body problem in $3N - 3$ dimensions, hyperspherical coordinates with one hyperradius and $3N - 4$ hyperangles are then applied to describe the system to better incorporate rotational symmetries. There are different ways of defining the hyperangles and they could again be represented by binary trees graphically. There will be $3N - 3$ leaves, corresponding to $3N - 3$ Cartesian components, and $3N - 4$ nodes, to which accord the hyperangles. By convention, a line joining a node $\gamma_i$ from the left represents a multiplicative factor of $\cos \gamma_i$ and that from the right $\sin \gamma_i$. Then the definition of each Cartesian component could be read directly from the tree representation by simply tracing from its corresponding leaf to the root. The range of each hyperangle ought to be specified as well - $[0, 2\pi)$ if its branches contain no more nodes, $[0, \pi]$ if there are further nodes attaching to one of the branches, and $[0, \pi/2]$ if both branches contain further nodes \cite{Aqui1986, Aqui2004}. Two different trees associated with $N=3$ are displayed in Fig.~\ref{fig:tree_1}. The hyperspherical coordinates are given by 

    \begin{figure}[H]
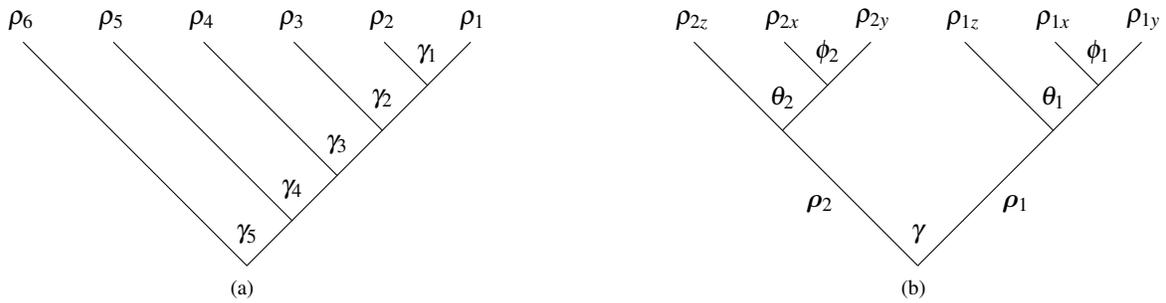

        \centering
        \subfigure[]{
            \begin{istgame}
                \setistgrowdirection'{north}
                \setistNullNodeStyle{0pt}
                \xtdistance{6mm}{12mm}
                \istroot(0)[null node]<[yshift=5pt]>{$\gamma_5$}
                    \istb \istb
                \endist
                \istroot(1)(0-1)[null node]{}
                    \istb \istbm
                \endist
                \istroot(2)(0-2)[null node]<[yshift=5pt]>{$\gamma_4$}
                    \istb \istb
                \endist
                \istroot(11)(1-1)[null node]{}
                    \istb \istbm
                \endist
                \istroot(21)(2-1)[null node]{}
                    \istb \istbm
                \endist
                \istroot(22)(2-2)[null node]<[yshift=5pt]>{$\gamma_3$}
                    \istb \istb
                \endist
                \istroot(111)(11-1)[null node]{}
                    \istb \istbm
                \endist
                \istroot(211)(21-1)[null node]{}
                    \istb \istbm
                \endist
                \istroot(221)(22-1)[null node]{}
                    \istb \istbm
                \endist
                \istroot(222)(22-2)[null node]<[yshift=5pt]>{$\gamma_2$}
                    \istb \istb
                \endist
                \istroot(1111)(111-1)[null node]{}
                    \istb{}[al]{\rho_6} \istbm
                \endist
                \istroot(2111)(211-1)[null node]{}
                    \istb{}[al]{\rho_5} \istbm
                \endist
                \istroot(2211)(221-1)[null node]{}
                    \istb{}[al]{\rho_4} \istbm
                \endist
                \istroot(2221)(222-1)[null node]{}
                    \istb{}[al]{\rho_3} \istbm
                \endist
                \istroot(2222)(222-2)[null node]<[yshift=5pt]>{$\gamma_1$}
                    \istb{}[al]{\rho_2} \istb{}[al]{\rho_1}
                \endist
            \end{istgame}
        }
        \hspace{2cm}
        \subfigure[]{
            \begin{istgame}
                \setistgrowdirection'{north}
                \setistNullNodeStyle{0pt}
                \xtdistance{6mm}{12mm}
                \istroot(0)[null node]<[yshift=5pt]>{$\gamma$}
                    \istb \istb
                \endist
                \istroot(1)(0-1)[null node]<[xshift=-20pt]>{$\boldsymbol{\rho}_2$}
                    \istb \istbm
                \endist
                \istroot(2)(0-2)[null node]<[xshift=20pt]>{$\boldsymbol{\rho}_1$}
                    \istbm \istb
                \endist
                \istroot(11)(1-1)[null node]{}
                    \istb \istbm
                \endist
                \istroot(22)(2-2)[null node]{}
                    \istbm \istb
                \endist
                \istroot(111)(11-1)[null node]<[yshift=5pt]>{$\theta_2$}
                    \istb \istb
                \endist
                \istroot(222)(22-2)[null node]<[yshift=5pt]>{$\theta_1$}
                    \istb \istb
                \endist
                \istroot(1111)(111-1)[null node]{}
                    \istb{}[al]{\rho_{2z}} \istbm
                \endist
                \istroot(1112)(111-2)[null node]<[yshift=5pt]>{$\phi_2$}
                    \istb{}[al]{\rho_{2x}} \istb{}[al]{\rho_{2y}}
                \endist
                \istroot(2221)(222-1)[null node]{}
                    \istb{}[al]{\rho_{1z}} \istbm
                \endist
                \istroot(2222)(222-2)[null node]<[yshift=5pt]>{$\phi_1$}
                    \istb{}[al]{\rho_{1x}} \istb{}[al]{\rho_{1y}}
                \endist
            \end{istgame}
        }
        \caption{Two definitions of hyperangles for $N = 3$. (a) is a straightforward extension of spherical coordinates, and (b) is defined such that $\boldsymbol{\rho}_1$ and $\boldsymbol{\rho}_2$ are very much living in their own 3-dimensional space.}
        \label{fig:tree_1}
    \end{figure}

    \begin{equation}
        \boldsymbol{\rho}_a = \begin{pmatrix}
            \rho_1 \\ \rho_2 \\ \rho_3 \\ \rho_4 \\ \rho_5 \\ \rho_6
        \end{pmatrix} = \begin{pmatrix}
            \rho \sin \gamma_1 \sin \gamma_2 \sin \gamma_3 \sin \gamma_4 \sin \gamma_5 \\
            \rho \cos \gamma_1 \sin \gamma_2 \sin \gamma_3 \sin \gamma_4 \sin \gamma_5 \\
            \rho \cos \gamma_2 \sin \gamma_3 \sin \gamma_4 \sin \gamma_5 \\
            \rho \cos \gamma_3 \sin \gamma_4 \sin \gamma_5 \\
            \rho \cos \gamma_4 \sin \gamma_5 \\
            \rho \cos \gamma_5
        \end{pmatrix},
        \end{equation}
    for the tree in panel (a), whereas the tree of panel (b) corresponds to
        \begin{equation}
        \boldsymbol{\rho}_b = \begin{pmatrix}
            \boldsymbol{\rho}_1 \\ \boldsymbol{\rho}_2
        \end{pmatrix} = \begin{pmatrix}
            \rho_{1,y} \\ \rho_{1,x} \\ \rho_{1,z} \\ \rho_{2,y} \\ \rho_{2,x} \\ \rho_{2,z}
        \end{pmatrix} = \begin{pmatrix}
            \rho \sin \gamma \sin \theta_1 \sin \phi_1 \\
            \rho \sin \gamma \sin \theta_1 \cos \phi_1 \\
            \rho \sin \gamma \cos \theta_1 \\
            \rho \cos \gamma \sin \theta_2 \sin \phi_2 \\
            \rho \cos \gamma \sin \theta_2 \cos \phi_2 \\
            \rho \cos \gamma \cos \theta_2
        \end{pmatrix}.
    \end{equation}

    Furthermore, it is worth noting that the tree shown in panel (b) of Fig. \ref{fig:tree_1} could be obtained by attaching a ``spherical fork" - the tree representation for regular spherical coordinates - to every virtual node of a Jacobi tree after removing the physical leaves for $N = 3$. Such trees can be constructed for any $N$ and have an intuitive and physical meaning of characterizing the motion of each virtual body in 3-dimensional spherical coordinates and then {\it joining} them as projections of higher-dimensional vectors at each node. Its usefulness will be further demonstrated in Part \ref{L}.

\section{Grand Angular Momentum} \label{grandL}

We have discussed the usefulness of Jacobi coordinates and their representations in hyperspherical coordinates in reducing the N-body setup to a single virtual body in a $(3N-3)$-dimensional space. Next, we study the generalization of angular momentum to the N-body problem. First, the original definition of angular momentum needs to be extended to higher dimensions and, therefore, become a tensor. It is called the grand angular momentum tensor in the literature and is defined as
    \begin{equation} 
        \boldsymbol{\Lambda} = \boldsymbol{\rho} \wedge \boldsymbol{P}. 
    \end{equation}
and
     \begin{equation}
        \Lambda_{ij} = \rho_i P_j - \rho_j P_i. 
    \end{equation}
component-wise. Its magnitude is given by
     \begin{equation}
        \Lambda^2 = \frac{1}{2} \sum_{i,j} (\Lambda_{ij})^2.
    \end{equation}

   In hyperspherical coordinates, $\boldsymbol{\rho} = \rho \boldsymbol{\hat{\rho}}$, the momentum is then 
    \begin{equation}
        \boldsymbol{P} = \mu \boldsymbol{\dot{\rho}}=\mu (\dot{\rho} \boldsymbol{\hat{\rho}} + \rho \boldsymbol{\dot{\hat{\rho}}}),
    \end{equation}
    and $\Lambda^2$ can be calculated with the Lagrange identity as follows
    \begin{align} \label{gam}
        \Lambda^2 &= (\boldsymbol{\rho} \wedge \mu {\dot{\boldsymbol{\rho}}})^2 = \mu^2 (\rho \boldsymbol{\hat{\rho}} \wedge (\dot{\rho} \boldsymbol{\hat{\rho}} + \rho \boldsymbol{\dot{\hat{\rho}}}))^2 = \mu^2 \rho^4 (\boldsymbol{\hat{\rho}} \wedge \boldsymbol{\dot{\hat{\rho}}})^2 \nonumber \\
        &= \mu^2 \rho^4 [(\boldsymbol{\hat{\rho}} \cdot \boldsymbol{\hat{\rho}}) (\boldsymbol{\dot{\hat{\rho}}} \cdot \boldsymbol{\dot{\hat{\rho}}}) - (\boldsymbol{\hat{\rho}} \cdot \boldsymbol{\dot{\hat{\rho}}})^2] = \mu^2 \rho^4 \boldsymbol{\dot{\hat{\rho}}} \cdot \boldsymbol{\dot{\hat{\rho}}}.
    \end{align}    
    It could be observed that the dependence of $\Lambda^2$ on $\rho$ and the hyperangles are separated, showcasing the advantage of hyperspherical formulation. One can also see that it recovers the two-body result where $\boldsymbol{\dot{\hat{\rho}}} = \dot{\theta}$ and $L = \mu \rho^2 \dot{\theta}^2$.

    Further, the exact dependence of $\boldsymbol{\dot{\hat{\rho}}} \cdot \boldsymbol{\dot{\hat{\rho}}}$ on the hyperangles is directly readable from the hyperspherical trees.

    \begin{theorem} \label{thm1}
        In the tree representation of hyperspherical coordinates, for every hyperangle $\gamma_i$, tracing its path to the root, label the jth node it reaches from the right $\alpha_{ij}$, and the kth node reached from the left $\beta_{ik}$, then $\boldsymbol{\dot{\hat{\rho}}} \cdot \boldsymbol{\dot{\hat{\rho}}} = \sum_i \dot{\gamma}_i^2 (\prod_j \sin^2 \alpha_{ij}) (\prod_k \cos^2 \beta_{ik})$.
    \end{theorem}

    The proof is by induction.

    \begin{proof}
    
        For the simplest tree possible with only the root and two leaves, the theorem holds true. Explicitly,
        \begin{equation*}
            \boldsymbol{\hat{\rho}} =  \begin{pmatrix}
                \sin \gamma \\
                \cos \gamma
            \end{pmatrix},
            \hspace{2cm}
            \boldsymbol{\dot{\hat{\rho}}} =  \begin{pmatrix}
                \dot{\gamma} \cos \gamma \\
                - \dot{\gamma} \sin \gamma
            \end{pmatrix},
            \hspace{2cm}
            \boldsymbol{\dot{\hat{\rho}}} \cdot \boldsymbol{\dot{\hat{\rho}}} = \dot{\gamma}^2.
        \end{equation*}

        Assume the theorem holds true for trees with all $N - 1$ nodes on the right branch and add the $N$th node as the new root joining the old tree (now as a branch) with a bare branch on the left, such as shown in FIG. \ref{fig:tree_1}a, then
        \begin{equation*}
            \boldsymbol{\hat{\rho}}_{N+1} =  \begin{pmatrix}
                \boldsymbol{\hat{\rho}}_N \sin \alpha_N \\
                \cos \alpha_N
            \end{pmatrix},
            \hspace{2cm}
            \boldsymbol{\dot{\hat{\rho}}}_{N+1} =  \begin{pmatrix}
                \dot{\alpha}_N \boldsymbol{\hat{\rho}}_{N} \cos \alpha_N + \boldsymbol{\dot{\hat{\rho}}}_N \sin \alpha_N \\
                - \dot{\alpha}_N \sin \alpha_N
            \end{pmatrix},
        \end{equation*}
        and
        \begin{equation*}
             \boldsymbol{\dot{\hat{\rho}}}_{N+1} \cdot \boldsymbol{\dot{\hat{\rho}}}_{N+1} = \dot{\alpha_N}^2 + \boldsymbol{\dot{\hat{\rho}}}_N \cdot \boldsymbol{\dot{\hat{\rho}}}_N \hspace{1mm }\sin^2 \alpha_N,
        \end{equation*}
        which has the desired form.
    
        The same applies to trees with nodes all on the left branch, only this time
        \begin{equation*}
            \boldsymbol{\hat{\rho}}_{N+1} =  \begin{pmatrix}
                \sin \beta_N \\
                \boldsymbol{\hat{\rho}}_N \cos \beta_N \\
            \end{pmatrix},
            \hspace{2cm}
            \boldsymbol{\dot{\hat{\rho}}}_{N+1} =  \begin{pmatrix}
                \dot{\beta}_N \cos \beta_N \\
                - \dot{\beta}_N \boldsymbol{\hat{\rho}}_{N} \sin \beta_N + \boldsymbol{\dot{\hat{\rho}}}_N \hspace{1mm} \cos \beta_N
            \end{pmatrix},
        \end{equation*}
        and
        \begin{equation*}
             \boldsymbol{\dot{\hat{\rho}}}_{N+1} \cdot \boldsymbol{\dot{\hat{\rho}}}_{N+1} = \dot{\beta_N}^2 + \boldsymbol{\dot{\hat{\rho}}}_N \cdot \boldsymbol{\dot{\hat{\rho}}}_N  \hspace{1mm} \cos^2 \beta_N.
        \end{equation*}
    
        Next, for trees with nodes on both main branches - $M - 1$ on the right and $N - 1$ on the left - and they are joined by the root $\gamma$ form a tree for hyperspherical coordinates with $M+N$ leaves,
        \begin{equation*}
            \boldsymbol{\hat{\rho}}_{M+N} =  \begin{pmatrix}
                \boldsymbol{\hat{\rho}}_R \sin \gamma \\
                \boldsymbol{\hat{\rho}}_L \cos \gamma
            \end{pmatrix},
            \hspace{2cm}
            \boldsymbol{\dot{\hat{\rho}}}_{M+N} =  \begin{pmatrix}
                \dot{\gamma} \boldsymbol{\hat{\rho}}_R \cos \gamma + \boldsymbol{\dot{\hat{\rho}}}_R \sin \gamma \\
                - \dot{\gamma} \boldsymbol{\hat{\rho}}_L \sin \gamma + \boldsymbol{\dot{\hat{\rho}}}_L \cos \gamma
            \end{pmatrix},
        \end{equation*}
        and
        \begin{equation*}
             \boldsymbol{\dot{\hat{\rho}}}_{M+N} \cdot \boldsymbol{\dot{\hat{\rho}}}_{M+N} = \dot{\gamma}^2 + \boldsymbol{\dot{\hat{\rho}}}_R \cdot \boldsymbol{\dot{\hat{\rho}}}_R  \hspace{1mm} \sin^2 \gamma + \boldsymbol{\dot{\hat{\rho}}}_L \cdot \boldsymbol{\dot{\hat{\rho}}}_L \hspace{1mm} \cos^2 \gamma,
        \end{equation*}
        which, again, has the form given by the theorem. Lastly, the same proof applies to two trees of any configuration joining together through a new root.
        
    \end{proof}

\subsection{Angular momentum} \label{L}

Theorem 1 applies to all possible trees. For instance, when applied to panel (b) of Fig. \ref{fig:tree_1}, we find 

    \begin{align*}
        \Lambda^2 &= \mu^2 \rho^4 (\dot{\alpha}^2 + (\dot{\theta}_1^2 + \dot{\phi}_1^2 \sin^2 \theta_1) \sin^2 \alpha + (\dot{\theta}_2^2 +  \dot{\phi}_2^2 \sin^2 \theta_2) \cos^2 \alpha) \notag \\
        &= \mu^2 \rho^4 (\frac{d}{dt} \arctan \frac{\rho_1}{\rho_2})^2 + \mu^2 (\rho_1 \csc \alpha)^4 (\dot{\theta}_1^2 + \dot{\phi}_1^2 \sin^2 \theta_1) \sin^2 \alpha + \\
        & \hspace{0.4cm} \mu^2 (\rho_2 \sec \alpha)^4 (\dot{\theta}_2^2 + \dot{\phi}_2^2 \sin^2 \theta_2) \cos^2 \alpha \notag \\
        &= \mu^2 (\dot{\rho}_1 \rho_2 - \rho_1 \dot{\rho}_2)^2 + L_1^2 \csc^2 \alpha + L_2^2 \sec^2 \alpha = L_{2,1}^2 + L_1^2 \csc^2 \alpha + L_2^2 \sec^2 \alpha,
    \end{align*}
    after defining
    \begin{equation*}
        L_i^2 = \mu^2\rho_i^4(\dot{\theta_i}^2+\dot{\phi}_i^2 \sin^2 \theta_i),
    \end{equation*}
    as the angular momentum for the two virtual bodies and
    \begin{equation*}
        \boldsymbol{\rho}_{2,1} = \begin{pmatrix}
            \rho_1 \\ \rho_2 \\ 0
        \end{pmatrix},
        \hspace{1cm}
        \boldsymbol{P}_{2,1} = \mu \boldsymbol{\dot{\rho}}_{2,1}, 
        \hspace{1cm} 
        \boldsymbol{L}_{2,1} = \boldsymbol{\rho}_{2,1} \wedge \boldsymbol{P}_{2,1},
    \end{equation*}
    as the relative angular momentum between them. Therefore, the magnitude of the grand angular momentum tensor in six dimensions can be expressed in terms of angular momenta defined in the three-dimensional space. This result can be generalized, as the following theorem shows:
    
    \begin{theorem} \label{thm2}
    
        In the tree representation of hyperspherical coordinates with every leaf part of a spherical fork, two sources of angular momentum could be identified:

        i) For each spherical fork labeled by $i$, it by itself contributes $L_i^2 = \mu^2 \rho_i^4 (\dot{\theta}_i^2 + \dot{\phi}_i^2 \sin^2 \theta_i)$.

        ii) For each additional node labeled by $\gamma$, connected to forks $\gamma_{l_1}$, $\gamma_{l_2}$, ..., $\gamma_{l_m}$ on its left branch and $\gamma_{r_1}$, $\gamma_{r_2}$, ..., $\gamma_{r_n}$ on its right, define 
        \begin{equation*}
            \rho_{\gamma L} = \sqrt{\sum_{k=1}^m \rho_{\gamma_{l_k}}^2} \hspace{0.5cm} and \hspace{0.5cm} \rho_{\gamma R} = \sqrt{\sum_{k=1}^n \rho_{\gamma_{r_k}}^2},
        \end{equation*}
        then such node by itself contributes $L_\gamma^2 = L_{\gamma L, \gamma R}^2 = \mu^2 (\dot{\rho}_{\gamma R} \rho_{\gamma L} - \rho_{\gamma R} \dot{\rho}_{\gamma L})^2$.

        Tracing their paths to the root, label $\alpha_{ip}$ the pth node reached from the right by the spherical fork $i$ and $\alpha_{\gamma p}$ the pth node reached from the right by the additional node $\gamma$, and similarly $\beta_{iq}$ and $\beta_{\gamma q}$ for the qth node reached from the left, $\Lambda^2$ can be decomposed as $\Lambda^2 = \sum_i L_i^2 (\prod_p \csc^2 \alpha_{ip}) (\prod_q \sec^2 \beta_{iq}) + \sum_\gamma L_\gamma^2 (\prod_p \csc^2 \alpha_{\gamma p}) (\prod_{q} \sec^2 \beta_{\gamma q})$.      

    \end{theorem}

    \begin{proof}

        i) is a direct consequence of Theorem \ref{thm1}.
        
        For ii), consider an additional node $\gamma$, there adds the new term $\mu^2 \rho_\gamma^4 \dot{\gamma}^2$ , and it has two branches - to the left $\rho_{\gamma L} = \rho_\gamma \cos \gamma$ and $\rho_{\gamma R} = \rho_\gamma \sin \gamma$ to the right, so $\gamma = \arctan (\rho_{\gamma R} / \rho_{\gamma L})$. The form of $L_\gamma$ is obtained by taking a time derivative. If a tree $\delta_R$ joins a node $\delta$ as the right branch, then the angular part in Eq.~\ref{gam} is now $\rho_\delta=\rho_{\delta R} \csc \delta$. Hence, the total scaling factor for $L_{\delta R}^2$ is $\csc^2 \delta$. From the left, the same applies, yielding $\sec^2 \delta$ as the multiplicative factor.
    \end{proof}

   For each $\gamma$, having $\csc^2 \gamma = \rho_\gamma^2 / \rho_{\gamma R}^2$ and $\sec^2 \gamma = \rho_\gamma^2 / \rho_{\gamma L}^2$, the problem is then transformed back to 3-dimensions. However, the reduction of dimensionality requires considering more angular momenta. There are $N - 1$ contributions from spherical forks, corresponding to the 3-dimensional angular momentum of the $N - 1$ virtual bodies. There are also $N - 2$ additional nodes, representing hyperangles not defined within spherical forks and having the exact structure of a Jacobi tree, describing angular momentum of relative motion between the virtual bodies. To see it more clearly, one can define $\boldsymbol{\rho}_\gamma$ and derive $\boldsymbol{P}_\gamma$ and $L_\gamma^2$ as
    \begin{equation*}
        \boldsymbol{\rho}_\gamma = \begin{pmatrix}
            \rho_{\gamma R} \\ \rho_{\gamma L} \\ 0
        \end{pmatrix},
        \hspace{1cm}
        \boldsymbol{P}_\gamma = \mu  \boldsymbol{\dot{\rho}}_\gamma,
        \hspace{1cm}
        \boldsymbol{L}_\gamma^2 = (\boldsymbol{\rho}_\gamma \wedge \boldsymbol{P}_\gamma)^2 = \mu^2 (\dot{\rho}_{\gamma R} \rho_{\gamma L} - \rho_{\gamma R} \dot{\rho}_{\gamma L})^2,
    \end{equation*}

    For example, the Jacobi tree in panel (b) of Fig. \ref{fig:4-body} can be extended to become a hyperspherical tree by removing the solid circles representing the physical bodies and appending to each virtual body a spherical fork, producing Fig. \ref{fig:3}, to which Theorem \ref{thm2} applies. Then, one can reads directly from it
    \begin{equation}
        \Lambda^2 = L_{32,1}^2 + L_1^2 \csc^2 \gamma_{32,1} + (L_{3,2}^2 + L_3^2 \sec^2 \gamma_{3,2} + L_2^2 \csc^2 \gamma_{3,2}) \sec^2 \gamma_{32,1}.
    \end{equation}
    The expression only contains quantities in 3-dimensional space, but is a sum of $2 \cross 4 - 3 = 5$ angular momenta with different scaling factors.

    It is also worth noting that one could obtain $L_i$'s and $L_\gamma$'s in the unweighted coordinates as well, further confirming the physical interpretation that they are the individual angular momenta of the virtual bodies and the relative angular momenta between them.

    In summary, it is possible to work with one grand angular momentum in $3N-3$ dimensions or with $2N-3$ angular momenta in 3-dimensional space.

    \begin{figure}
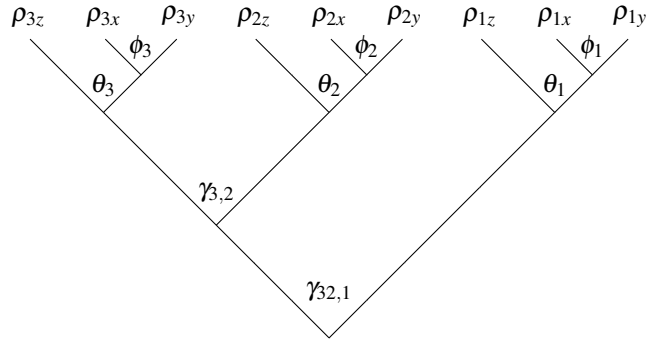

        \centering
            \begin{istgame}
                \setistgrowdirection'{north}
                \setistNullNodeStyle{0pt}
                \xtShowTerminalNodes
                \xtdistance{5mm}{10mm}
                \istroot(0)[null node]<[yshift=10pt]>{$\gamma_{32,1}$}
                    \istb \istb
                \endist
                \istroot(01)(0-1)[null node]{}
                    \istb \istbm
                \endist
                \istroot(02)(0-2)[null node]{}
                    \istbm \istb
                \endist
                \istroot(1)(01-1)[null node]{}
                    \istb \istbm
                \endist
                \xtdistance{10mm}{20mm}
                \istroot(2)(02-2)[null node]{}
                    \istbm \istb
                \endist
                \xtdistance{15mm}{30mm}
                \istroot(11)(1-1)[null node]<[yshift=5pt]>{$\gamma_{3,2}$}
                    \istb \istb
                \endist
                \xtdistance{10mm}{20mm}
                \istroot(22)(2-2)[null node]{}
                    \istbm \istb
                \endist
                \xtdistance{5mm}{10mm}
                \istroot(111)(11-1)[null node]<[yshift=3pt]>{$\theta_3$}
                    \istb \istb
                \endist
                \istroot(112)(11-2)[null node]<[yshift=3pt]>{$\theta_2$}
                    \istb \istb
                \endist
                \istroot(1111)(111-1)[null node]{}
                    \istb{}[al]{\rho_{3z}} \istbm
                \endist
                \istroot(1112)(111-2)[null node]<[yshift=3pt]>{$\phi_3$}
                    \istb{}[al]{\rho_{3x}} \istb{}[al]{\rho_{3y}}
                \endist
                \istroot(1121)(112-1)[null node]{}
                    \istb{}[al]{\rho_{2z}} \istbm
                \endist
                \istroot(1122)(112-2)[null node]<[yshift=3pt]>{$\phi_2$}
                    \istb{}[al]{\rho_{2x}} \istb{}[al]{\rho_{2y}}
                \endist
                \istroot(222)(22-2)[null node]<[yshift=3pt]>{$\theta_1$}
                    \istb \istb
                \endist
                \istroot(2221)(222-1)[null node]{}
                    \istb{}[al]{\rho_{1z}} \istbm
                \endist
                \istroot(2222)(222-2)[null node]<[yshift=3pt]>{$\phi_1$}
                    \istb{}[al]{\rho_{1x}} \istb{}[al]{\rho_{1y}}
                \endist
            \end{istgame}
        \caption{Hyperspherical tree associated with panel (c) Fig. \ref{fig:4-body}. $\boldsymbol{\rho}_1$ characterizes the motion of $\mu_{12,34}$, $\boldsymbol{\rho}_2$ that of $\mu_{3,4}$, and $\boldsymbol{\rho}_3$ that of $\mu_{1,2}$}
        \label{fig:3}
    \end{figure}

    \subsection{Scattering angle}

    It is well-known that angular momentum generates a repulsive barrier responsible for the stability of the motion, or in the case of scattering problems, it provides the reaction barrier for chemical processes. In the case of more than two bodies, the barrier appears as well in the dynamics, but it is produced due to the grand angular momentum as
    \begin{equation}
        E = \frac{\boldsymbol{P}^2}{2 \mu} + V(\boldsymbol{\rho}) = \frac{1}{2} \mu \dot{\rho}^2 + \frac{\Lambda^2}{2 \mu \rho^2} + V(\boldsymbol{\rho}),
    \end{equation}
    and it has the same effect of avoiding any close approach $\rho \rightarrow 0$ when $\Lambda^2 \neq 0$.
    
    The impact parameter $\boldsymbol{b}$ essential in all calculations of scattering and reactive processes \cite{Child1974} also deserves a generalization - as the projection of initial position $\boldsymbol{\rho}_0$ onto the hyperplane perpendicular to the initial momentum vector $\boldsymbol{P}_0$. For $N = 2$, $\boldsymbol{b}$ is reduced to a single number whereas for general $N$, $\boldsymbol{b}$ is a vector living in a $(3N - 4)$-dimensional hyperplane \cite{PR2020, Greene2017}. With this definition of $\boldsymbol{b}$, the initial grand angular momentum tensor has initial magnitude \cite{PR2021}
    \begin{equation}
         \Lambda_0^2 = (\boldsymbol{\rho}_0 \wedge \boldsymbol{P}_0)^2 = (\boldsymbol{b} \wedge \boldsymbol{P}_0)^2 = \boldsymbol{b}^2 \boldsymbol{P}_0^2 - (\boldsymbol{b} \cdot \boldsymbol{P}_0)^2 = 2\mu E b^2.
    \end{equation}

    When $V(\boldsymbol{\rho})=V(\rho)$ is independent of the hyperangles, or when the angular degrees of freedom can be integrated over to produce an effective potential
    \begin{equation*}
        V_{eff}(\rho)=\frac{\int V(\rho,\boldsymbol{\Omega})d\boldsymbol{\Omega}}{\int d\boldsymbol{\Omega}},
    \end{equation*}
    where $\boldsymbol{\Omega}$ represents all angular degrees of freedom, $\Lambda^2$ can be regarded as effectively conserved with magnitude $2\mu E b^2$. Then, from the above relations, we find
    \begin{equation} \label{diffeq}
        \boldsymbol{\hat{\rho}}' \cdot \boldsymbol{\hat{\rho}}' = \frac{d\boldsymbol{\hat{\rho}}}{d\rho} \cdot \frac{d\boldsymbol{\hat{\rho}}}{d\rho} = \frac{\boldsymbol{\dot{\hat{\rho}}} \cdot \boldsymbol{\dot{\hat{\rho}}}}{\dot{\rho}^2} = \frac{b^2}{\rho^4 (1-b^2/\rho^2 - V(\rho)/E)},
    \end{equation}
    where $\boldsymbol{\hat{\rho}}' \cdot \boldsymbol{\hat{\rho}}'$ could be read directly from the tree as stated in Theorem \ref{thm1} with a simple replacement of $d/dt$ by $d/d\rho$.

    Similarly, it is possible to transform the above expression using the chain rule of second derivatives to arrive at a second-order expression
    \begin{equation} \label{2nd-orderdiffeq}
        \boldsymbol{\hat{\rho}} \cdot \frac{d^2 \boldsymbol{\hat{\rho}}}{d\rho^2} = - \frac{b^2}{\rho^4 (1-b^2/\rho^2 - V(\rho)/E)}.
    \end{equation}
   Eq.~(\ref{diffeq}) resembles the expression for the scattering angle in the two-body problem. Hence, it can be interpreted as the scattering angle equation for the N-body problem. 
    
\section{Conclusion} \label{concl}

Using hyperspherical coordinates, we have developed a graphical approach for treating classical grand angular momentum in the N-body problem. Similarly, we have been able to demonstrate the relationship between individual and relative angular momenta of the bodies and the grand angular momentum $\Lambda$, which is essential to understand the N-body problem. Specifically, it has been shown that the exact form of $\Lambda^2$ is directly readable from the associated graphics, that the magnitude of higher-dimensional grand angular momentum can be decomposed into magnitudes of three-dimensional angular momenta. In deriving these results, the only assumption made on the potential is that it depends only on interparticle distances, making them widely applicable and a valuable asset in analysis. 

At the same time, our results have led us to a formulation of the scattering angle equation in Eq.~\ref{diffeq} for the N-body problem as a generalization of the well-known expression in two-body scattering. This result establishes a link between our knowledge of two-body physics and how it could be generalizable to the N-body problem. 
    
\backmatter

\bmhead{Acknowledgments}

The authors acknowledge the support of the Office of the Vice President for Research (OVPR) at Stony Brook University through the seed grant program.

\section*{Declarations}

\begin{itemize}


\item Data availability 

No data were created or analyzed in this study.

\item Author contribution

Zhongqi Liang: Conceptualization (equal); Formal analysis (lead); Investigation (equal); Visualization (lead); Writing – original draft (equal); Writing – review \& editing (equal). Jes\'us P\'erez-R\'ios: Conceptualization (equal); Investigation (equal); Supervision (lead).

\end{itemize}

\bibliography{bibliography}

\end{document}